\documentclass[a4paper,UKenglish,cleveref, autoref, thm-restate, total={6.3 in, 9.7 in}]{lipics-v2021}

\usepackage{todonotes}
\usepackage{color,colortbl}

\usepackage{mathtools, thm-restate}
\DeclarePairedDelimiter\norm{\lVert}{\rVert}

\usepackage{mathtools}

\usepackage[h]{esvect}

\usepackage{array}

\usepackage{listings}
\usepackage{microtype}

\usepackage[ruled, vlined]{algorithm2e}

\newcommand{\E}{\mathbb{E}}

\newcommand{\N}{\mathbb{N}}

\newcommand{\R}{\mathbb{R}}
\newcommand{\Z}{\mathbb{Z}}

\newcommand{\T}{\mathbb{T}}

\renewcommand{\epsilon}{\ensuremath\varepsilon}

\renewcommand{\phi}{\ensuremath{\varphi}}

\makeatletter
\newcommand\footnoteref[1]{\protected@xdef\@thefnmark{\ref{#1}}\@footnotemark}
\makeatother

\usepackage{comment}

\newcommand{\ex}[1]{\mathbb{E} \left[ #1 \right]}

\newcommand{\pr}[1]{\mathbb{P} \left[ #1 \right]}

\title{Sublinear Cuts are the Exception in BDF-GIRGs}

\author{Marc Kaufmann}{ETH Zurich, Switzerland}{marc.kaufmann@inf.ethz.ch}{}{Swiss National Science Foundation [grant number 200021\_192079]}
\author{Raghu Raman Ravi}{ETH Zurich, Switzerland}{raravi@student.ethz.ch}{}{}
\author{Ulysse Schaller}{ETH Zurich, Switzerland}{ulysse.schaller@inf.ethz.ch}{}{Swiss National Science Foundation [grant number 200021\_192079]}

\authorrunning{M. Kaufmann, R.\,R. Ravi, and U. Schaller}

\Copyright{Marc Kaufmann, Raghu Raman Ravi, and Ulysse Schaller}

\ccsdesc[100]{} 

\keywords{Real-world Networks, Geometric Inhomogeneous Random Graphs, Boolean Distance Functions, Network Robustness, Small Separators, Sparse Cuts, Clustering} 

\category{} 

\relatedversion{} 

\ccsdesc[500]{Theory of computation~Random network models}
\ccsdesc[500]{Theory of computation~Generating random combinatorial structures}
\ccsdesc[500]{Mathematics of computing~Random graphs}

\ArticleNo{23}

\begin{document}

\maketitle

\begin{abstract}
    The introduction of geometry has proven instrumental in the efforts towards more realistic models for real-world networks. In Geometric Inhomogeneous Random Graphs (GIRGs), Euclidean Geometry induces clustering of the vertices, which is widely observed in networks in the wild. Euclidean Geometry in multiple dimensions however restricts proximity of vertices to those cases where vertices are close in each coordinate. 
    We introduce a large class of GIRG extensions, called BDF-GIRGs, which capture arbitrary hierarchies of the coordinates within the distance function of the vertex feature space. These distance functions have the potential to allow more realistic modeling of the complex formation of social ties in real-world networks, where similarities between people lead to connections. Here, similarity with respect to certain features, such as familial kinship or a shared workplace, suffices for the formation of ties. It is known that - while many key properties of GIRGs, such as log-log average distance and sparsity, are independent of the distance function - the Euclidean metric induces small separators, i.e.\ sublinear cuts of the unique giant component in GIRGs, whereas no such sublinear separators exist under the component-wise minimum distance. Building on work of Lengler and Todorovi\'{c}, we give a complete classification for the existence of small separators in BDF-GIRGs. We further show that BDF-GIRGs all fulfill a stochastic triangle inequality and thus also exhibit clustering. 
\end{abstract}

\section{Introduction}

Bringing generative graph models closer to applications has driven network science since its inception. This includes the design of models which capture structural properties widespread among real networks. A prominent such model are Geometric Inhomogeneous Random Graphs (GIRGs), which are known to be sparse, small worlds, and whose degrees follow a power-law distribution \cite{bringmann2016average}.  Recent research has shown that they are well-suited to model geometric network features such as the local clustering coefficient as well as closeness and betweenness centrality of real networks~\cite{dayan2023expressivity}. In GIRGs, each vertex comes equipped with a set of coordinates in a geometric ground space and a weight, both drawn independently. Pairs of vertices $u,v$ are then connected independently with a probability which depends on the product of the vertex weights $w_u$ and $w_v$, and decays as their distance in the ground space increases. 
The role played by the distance function has been relatively unexplored. While we know that many graph properties are invariant under this choice~\cite{bringmann2016average}, the robustness of a GIRG, that is, how much of the graph's giant component can be removed before it falls apart (into chunks of comparable size) crucially depends on it~\cite{lengler2017existence}. Questions of robustness and fragility such as this one have been widely researched both for random graph models and real-world networks, in the quest of properties that are universal across networks. Analytical and numerical evidence suggests that there is a trade-off between robustness and performance, such that both cannot be simultaneously optimized~\cite{pasqualetti2020fragility}. A further central paradigm is the controversial "robust yet fragile nature of the internet"~\cite{doyle2005robust, hasheminezhad2023myth}. Robustness can be examined with respect to vertex or edge removal - and removals can be adversarial or random. Understanding which removal strategies are most successful has developed into a research direction of its own~\cite{bellingeri2018analyses, iyer2013attack, tian2017articulation, wandelt2018comparative}. 

In our present work, we consider robustness with respect to adversarial edge removal. In GIRGs, choosing the Euclidean metric induces a graph which contains small separators, that is, edge cuts of sublinear size which split the giant component into two connected components of linear size. Using instead the so-called minimum-component distance function (MCD), then in dimensions $d\ge 2$ - in dimension $d=1$, the two distance functions coincide - Lengler and Todorovi\'{c} have demonstrated that all sublinear separators disappear with high probability~\cite{lengler2017existence}. This coincides with the picture that we encounter in Erdös-Rényi random graphs at edge probability $p=\frac{1+ \Omega(1)}{n}$, as shown by Luczak and McDiarmid~\cite{luczak2001bisecting}. In both cases, the proofs proceed by a two-round exposure of the edges in the graph. After the first, larger, batch of edges is unveiled, the graph already contains a giant component - and this component contains only few sparse cuts. In the second round, one can show that enough edges are sprinkled between the vertices that are only sparsely connected, so that all sublinear cuts will disappear. Compared to Erdös-Rényi graphs, several complications arise. It is, in particular, not possible to sample the edges in two independent batches. Instead, one can produce two "almost independent" batches of edges by unveiling first a number of $d-1$ coordinates of the vertices and then unveiling their remaining $d$th coordinate. It is unavoidable that the giant component grows substantially, as a constant fraction of the total number of edges in the graph is unveiled this way, creating many new potential vertex bipartitions in the giant component which may yield sublinear cuts. This problem can be addressed by uncovering only coordinate information of bounded-weight vertices within the giant component for the second batch. In their elegant proof, Lengler and Todorovi\'{c} then show that edges incident to this vertex set do not increase the size of the giant by too much, using an Azuma-Hoeffding type estimate. 

As we will demonstrate, this procedure can be extended far beyond the case of MCD-GIRGs. We first generalize GIRGs to accommodate a large class of underlying distance functions, which we call Boolean Distance Functions (BDF). 
Intuitively speaking, this family of distance functions covers arbitrary combinations of minima and maxima of subsets of coordinates. For example, even if two individuals have exactly the same hobby (in this case, the minimum component distance is zero), they may still be unlikely to know each other if they live in different countries. A more refined notion could be for example, the formula $dist \coloneqq \min\{dist_{work}, \max\{dist_{hobby}, dist_{residence} \}\}$ which encodes that two individuals are likely to know each other if they either work in closely related fields, or if they have at the same time similar hobbies and locations. This illustrates the potential of BDF for building more realistic models for real-world networks. We remark that the dimension $d$ of the underlying geometric space is always assumed to be a constant with respect to the number of vertices $n$.

In many real-world instances, particularly social networks, connections are formed not based on an \textit{averaged} similarity as the one captured by Euclidean distances which weight all features in a symmetric way, but rather one some feature set that will naturally dominate. One may think of, in a social setting, an edge encoding that two people, i.e.\ vertices, know each other and closeness in a specific coordinate encoding having a parent in common -- or more generally how many generations one needs to go back to find a common ancestor. Even if all other features, e.g.\ education, age, wealth, domicile, interests, differ greatly, sharing a parent will usually ensure that two people know each other. 
Sometimes single features are not dominant enough, but being similar in a specific subset of $k$ features will render the other $d-k$ features obsolete, i.e.\ these $k$ features will suffice to guarantee a lower bound on the connection probability. One such example could be a feature that encodes that two people live in the same town with a population of ca.\ 100'000 people (which does not suffice for a connection) and another playing table tennis in a club (which on its own also is not sufficient). But it is plausible that playing in a table tennis club and living in the same town jointly ensures a constant lower bound on the connection probability. This is in line with the complex mechanisms underlying the observed formation of social ties in real networks, where similarities with regard to certain features are dominant but the baseline provided by similarity in this regard is combined with similarities or differences in other dimensions \cite{mcpherson2001birds}.
Such more sophisticated situations can be captured with BDF but are completely impossible to encode in the Euclidean setting - and, due to the equivalence of norms on finite-dimensional vector spaces - by any other metric induced by a norm. 

The precise definition of Boolean Distance Functions can be found in Definition~\ref{def:bdf}. In this family of GIRGs, there is only one subfamily, which we call Single-Coordinate Outer-Max (SCOM) GIRGs, that contains small separators. These collections of GIRGs exhibit a distance function that can be expressed as the maximum of one singled-out coordinate and an arbitrary "Boolean" combination of the other coordinates. This is our first main result:

\begin{theorem}
\label{thm:natural}
    Let $G$ be a GIRG induced by a SCOM BDF $\kappa$ acting on the $d$-dimensional torus $\T^d$. Then, with probability $1-o(1)$, the giant component of $G$ has a separator of size $o(n)$. 
\end{theorem}

We note that the separators can be described explicitly. One can consider two hyperplanes that are perpendicular to the singled-out coordinate axis, which bisect the ground space into two halves. It is then possible to show that a sublinear number of edges cuts across the hyperplanes, connecting two linear-sized components, hence yielding the natural small separators.

We then proceed to show - leveraging a modified version of the algorithm by Lengler and Todorovi\'{c}~\cite{lengler2017existence} - that for all other BDF-GIRGs, a two-round exposure of the vertex coordinates generates a graph where all potential sparse cuts in the giant component are erased in the second round. The two key modifications concern first the partition of the coordinates into batches. Here, the crucial ingredient is an upper bound on the distance function by a minimum of two distance functions which involve each only a subset of (disjoint) coordinates. This then enables an extension of the two-round exposure to arbitrary non-SCOM BDF-GIRGs, by allowing an (under)estimate of the edges contributed by the second exposure round. We also need to adjust the criteria for when edges are inserted. With these additional insights we are able to prove our second main result:

\begin{theorem}\label{thm:robust}
    Let $G$ be a GIRG induced by a non-SCOM BDF $\kappa$ acting on the $d$-dimensional torus $\T^d$. Then, with probability $1-o(1)$, the giant component of $G$ has no separator of size $o(n)$. 
\end{theorem}

Together, Theorems~\ref{thm:natural} and~\ref{thm:robust} provide a complete characterization of the occurrence of sublinear separators in BDF-GIRGs. Finally, we show that all BDFs satisfy a stochastic version of the triangle inequality. From this, it immediately follows that all BDF-GIRGs exhibit clustering. 

\begin{theorem}\label{thm:clustering}
    Let $G$ be a GIRG induced by a BDF $\kappa$ acting on the $d$-dimensional torus $\T^d$. Then, with probability $1-o(1)$, its clustering coefficient is constant, i.e.\ $\textsc{cc}(G) = \Theta(1)$. 
\end{theorem}
In social networks, the clustering coefficient has a natural interpretation: What proportion of the friends of a node are also mutual friends ? In sparse networks, a pair of nodes $u,v$ having a common friend thus greatly boosts their chances of being directly connected. Recently it has further been shown that, in a different setting, namely for GIRGs whose distance function is induced by an $L_p$-norm, the clustering coefficient can be used to estimate the dimension of the underlying ground space~\cite{friedrich2023simple}. The precise definition of $\textsc{cc}(G)$ and the proof of Theorem~\ref{thm:clustering} can be found in Appendix \ref{sec:clustering}. In the rest of the paper, we will say that an event happens \emph{with high probability (w.h.p.)} if it happens with probability $1-o(1)$.

\section{Geometric Inhomogeneous Random Graphs and Underlying Geometric Spaces}

In this section, we define Boolean Distance Functions (BDFs) as well as Geometric Inhomogeneous Random Graphs (GIRGs). At the end of the section we give some useful properties of GIRGs and two more general lemmata that we will be needed in the proofs.

\subsection{Definition of Boolean Distance Functions}

We start by introducing the notion of Boolean Distance Functions. They define a symmetric and translation-invariant distance between any two points in the $d$-dimensional torus $\T^d \coloneqq \R^d/\Z^d$, and hence can be characterized by an even function $\kappa : \T^d \rightarrow \R_{\geq 0}$.

The distance $dist(x,y)$ between two points $x,y\in\T^d$ is then given by $dist(x,y)\coloneqq \kappa(x-y)$. For the sake of simplicity, all the properties of the distance $dist(\cdot,\cdot)$ will be expressed in terms of $\kappa(\cdot)$. Note that on the one-dimensional torus, the distance between two points $x,y\in[0,1)$ is given by
\begin{align*}
    |x - y|_T \coloneqq \min(|x - y|, 1 - |x - y|).
\end{align*}
A Boolean Distance Function is then defined recursively as follows.

\begin{definition}
\label{def:bdf}
Let $d\in\N$ be a positive integer, and let $\kappa : \T^d \rightarrow \R_{\geq 0}$ and let $x = (x_1, x_2, .., x_d) \in \T^d$ be an arbitrary point. Then $\kappa$ is a \emph{Boolean Distance Function (BDF)} if:
\begin{itemize}
    \item When $d = 1$, then $\kappa(x) = |x|_T$.
    
    \item When $d \ge 2$, then there exists a non-empty proper subset $S \subsetneq [d]$ of coordinates such that
    \begin{align*}
        \kappa(x) = \max(\kappa_1((x_i)_{i \in S}), \kappa_2((x_i)_{i \notin S})) \quad \textnormal{or} \quad \kappa(x) = \min(\kappa_1((x_i)_{i \in S}), \kappa_2((x_i)_{i \notin S})),
    \end{align*}
    where $\kappa_1: \T^{|S|} \rightarrow \R_{\geq 0}$ and $\kappa_2: \T^{d-|S|} \rightarrow \R_{\geq 0}$ are Boolean Distance Functions.
\end{itemize}

In the $d\ge2$ case, we call $\kappa_1$ and  $\kappa_2$ the \emph{comprising functions} of $\kappa$. Moreover, we say that $\kappa$ is an \emph{outer-max} BDF if it is defined recursively as the maximum of two other BDFs, and we say $\kappa$ is an \emph{outer-min} BDF if it is defined recursively as the minimum of two other BDFs.
\end{definition}
Since we only work with the torus geometry in this paper, we will simply write $|x|$ instead of $|x|_T$.

Observe that the max-norm $\kappa(x) \coloneqq \max_{i\in [d]} |x_i|$ and the minimum component distance (MCD) $\kappa(x) \coloneqq \min_{i\in [d]} |x_i|$ are both BDFs.

Given a Boolean Distance Function $\kappa$, we write $B_{\kappa}^r(x) \coloneqq \{y\in\T^d : \kappa(x-y) < r\}$ for the ball centered at point $x$ of radius $r$ with respect to the distance function induced by $\kappa$. Moreover, we write $V_{\kappa}(r)$ for the volume (or Lebesgue measure) of that ball. We will sometimes drop the $\kappa$ index and just write $V(r)$ for that quantity when the underlying BDF is clear from context.

We now define the notion of depth of a Boolean Distance Function. As we will see in Lemma~\ref{prop:volume}, $V_{\kappa}(r)$ is characterized (as a function of $r$) by the depth of $\kappa$.

\begin{definition}\label{def:depth}
Let $\kappa : \T^d \rightarrow \R_{\geq 0}$ be a Boolean Distance Function. Then the \emph{depth} of $\kappa$, written $\mathcal{D}(\kappa)$, is defined recursively as follows:
\begin{itemize}
    \item When $d = 1$, then $\mathcal{D}(\kappa) \coloneqq 1$.
    
    \item When $d\ge2$, then $\mathcal{D}(\kappa) \coloneqq \mathcal{D}(\kappa_1) + \mathcal{D}(\kappa_2)$ if $\kappa$ is outer-max, and $\mathcal{D}(\kappa) \coloneqq \min(\mathcal{D}(\kappa_1),\mathcal{D}(\kappa_2))$ if $\kappa$ is outer-min, where $\kappa_1$ and $\kappa_2$ are the comprising functions of $\kappa$. 
\end{itemize}
\end{definition}

We now define the notion of Single-Coordinate Outer-Max (SCOM) Boolean Distance Function. These are BDFs that can be written as the maximum of a single coordinate and some other BDF acting on the remaining coordinates. As mentioned in the introduction, this property characterizes whether BDF-GIRGs have small separators or not.

\begin{definition}
Let $\kappa : \T^d \rightarrow \R_{\geq 0}$ be a Boolean Distance Function. We say that $\kappa$ is \emph{Single-Coordinate Outer-Max (SCOM)} if it can be written as
\begin{align*}
    \kappa(x) = \max(|x_k|, \kappa_0((x_i)_{i\neq k}))
\end{align*}
for some coordinate $k\in [d]$ and some BDF $\kappa_0: \T^{d-1} \rightarrow \R_{\geq 0}$. In dimension 1, we also say that $\kappa(x) = |x|$ is a SCOM BDF.
\end{definition}
Note that this means that the unique BDF acting on $d = 1$ dimension is SCOM.

\subsection{Definition of Geometric Inhomogeneous Random Graphs}

We now define Geometric Inhomogeneous Random Graphs (GIRGs). This model was initially introduced for the max-norm in~\cite{bringmann2019geometric}. Here, we use a version of this definition given in~\cite{lengler2017existence} that we generalize to cover any Boolean Distance Function as the underlying distance. Throughout the whole paper, we will consider undirected graphs with vertex set denoted by $\mathcal{V} \coloneqq [n]$ and edge set denoted by $\mathcal{E}$. Before we give the definition of GIRGs, we need to introduce the concept of (deterministically) power-law distributed weights. 
Intuitively, following a power law means that the proportion of vertices at a given weight $w$ decays as a polynomial in $w$.

\begin{definition}\label{def:power-law}
    Let $(w_v)_{v\in\mathcal{V}}$ be a sequence of weights associated with the vertex set, and let $\mathcal{V}_{\ge w} \coloneqq \{ v \in \mathcal{V} \mid w_v \geq w\}$ denote the set of vertices with weight at least $w$. We say that this sequence is \emph{power-law distributed with exponent $\beta$} if $w_v \ge 1$ for all $v\in\mathcal{V}$ and if the two following conditions are satisfied:
    \begin{enumerate}
        \item There exists some $\overline{w} = \overline{w}(n) \geq n^{\omega(1 / \log \log n)}$ such for all constant $\eta > 0$ there is a constant $c_1>0$ such that for all $1 \le w \le \overline{w}$,
        \begin{align}\label{eq:pl1}
            |\mathcal{V}_{\ge w}| \ge c_1 \frac{n}{w^{\beta - 1 + \eta}}. \tag{PL1}
        \end{align}

        \item For all constant $\eta > 0$ there is a constant $c_2>0$ such that for all $w \ge 1$,
        \begin{align}\label{eq:pl2}
            |\mathcal{V}_{\ge w}| \le c_2 \frac{n}{w^{\beta - 1 - \eta}}. \tag{PL2}
        \end{align}
    \end{enumerate}x
\end{definition}

Now, we are ready to define $\kappa$-GIRGs.
\begin{definition}\label{def:girg}
    Let $\beta\in (2,3)$, $\alpha>1$, $d\in\mathbb{N}$, and let $\kappa : \T^d \rightarrow \R_{\geq 0}$ be a Boolean Distance Function. Let $(w_v)_{v\in\mathcal{V}}$ be a sequence of weights that is power-law distributed with exponent $\beta$. A \emph{$\kappa$-Geometric Inhomogeneous Random Graph ($\kappa$-GIRG)} is obtained by the following two-step procedure:
    \begin{enumerate}
        \item Every vertex $v\in\mathcal{V}$ draws independently and uniformly at random\ a position $x_v$ in the torus $\T^d$.

        \item For every two distinct vertices $u,v \in\mathcal{V}$, we add an edge between $u$ and $v$ in $\mathcal{E}$ independently with some probability $p_{uv}$ satisfying
        \begin{align}\label{eq:girg-connection}
            c_L \cdot \min \Big\{ \frac{w_u w_v}{n \cdot V_{\kappa}(\kappa(x_u - x_v))} , 1 \Big\}^{\alpha} \le p_{uv} \le c_U \cdot \min \Big\{ \frac{w_u w_v}{n \cdot V_{\kappa}(\kappa(x_u - x_v))} , 1 \Big\}^{\alpha} \tag{EP}
        \end{align}
        for some constants $c_U \ge c_L > 0$.
    \end{enumerate}
\end{definition}

\subsection{Some useful GIRG properties}\label{sec:girg_properties}

We now recall some known properties of GIRGs that are crucially used for the proofs in later sections. We first remark that \cite{bringmann2016average} uses a very general definition of GIRG, and in particular our definition of $\kappa$-GIRG (Definition \ref{def:girg}) fits in their geometric setting. The first property is the existence and uniqueness of a connected component of linear size (i.e.\ containing $\Omega(n)$ vertices), which we call the \emph{giant component} of the graph. 

\begin{theorem}[Theorems 5.9 and 7.3 in~\cite{bringmann2016average}]
    \label{thm:giant}
    Let $\kappa : \T^d \rightarrow \R_{\geq 0}$ be a Boolean Distance Function and let $\mathcal{G} = (\mathcal{V}, \mathcal{E})$ be a $\kappa$-GIRG. There exists a constant $s_{max}>0$ such that w.h.p. $\mathcal{G}$ has a connected component of size at least $s_{max} n$. Additionally, w.h.p.\ all other connected components are at most poly-logarithmic in size, i.e., contain at most $\log^{O(1)} n $ vertices.
\end{theorem}

Another very useful result is the characterization of the degree distribution of GIRGs. The following lemma tells us that with high probability the degree distribution follows a power-law with the same exponent as the weight sequence.

\begin{lemma}[Theorems 6.3 and 7.3 in~\cite{bringmann2016average}]
    \label{lem:deg}
    Let $(\beta\in (2,3)$ and let $\kappa : \T^d \rightarrow \R_{\geq 0}$ be a Boolean Distance Function. Let $\mathcal{G} = (\mathcal{V}, \mathcal{E})$ be a $\kappa$-GIRG whose weight sequence $(w_v)_{v\in\mathcal{V}}$ follows a power-law with exponent $\beta$. Then the following holds:
    \begin{enumerate}
        \item For all $\eta > 0$, there exists $c_3 > 0$ such that w.h.p.\ for all $1 \le w \le \overline{w}$ (where $\overline{w}$ is the same as in Definition~\ref{def:power-law}),
        \begin{align*}
            |\{ v \in \mathcal{V} : \deg(v) \geq w\}| \ge c_3 \frac{n} {w^{\beta - 1 + \eta}}.
        \end{align*}

        \item For all $\eta > 0$, there exists $c_4 > 0$ such that w.h.p.\ for all $w \ge 1$,
        \begin{align*}
            |\{ v \in \mathcal{V} : \deg(v) \geq w\}| \le c_4 \frac{n} {w^{\beta - 1 - \eta}}.
        \end{align*}

        \item There exists $C>0$ such that for all $v\in\mathcal{V}$ we have $\E[\deg(v)] \le C w_v$, and moreover with probability $1-n^{-\omega(1)}$ we have $\deg(v) \le C \cdot (w_v + \log^2 n)$ for all $v\in\mathcal{V}$.
        
    \end{enumerate}
\end{lemma}

\section{Properties of Boolean Distance Functions}\label{sec:bdf}

In this section, we provide a few propositions about Boolean Distance Functions together with the main proof ideas. The complete proofs are given in Appendix \ref{sec:proof_bdf}. These propositions will be used for deriving our main results, but we believe they are also interesting results on their own. Remember that $V_{\kappa}(r)$ denotes the volume of a ball of radius $r$, with the distances measured with respect to $\kappa$. We start by analyzing how $V_{\kappa}(r)$ behaves as $r\rightarrow0$. When $\kappa : \T^d \rightarrow \R_{\geq 0}$ is the max-norm we have $V_{\kappa}(r)= \Theta(r^d)$, while for the MCD $\kappa(x) = \min_{i\in [d]} |x_i| $ we have $V_{\kappa}(r)= \Theta(r)$. The following proposition generalizes this to arbitrary BDFs.

\begin{proposition}\label{prop:volume}
	 Let $\kappa : \T^d \rightarrow \R_{\geq 0}$ be a Boolean Distance Function of depth $\mathcal{D}(\kappa)$ (see Definition~\ref{def:depth}). Then $V_{\kappa}(r) = \Theta(r^{\mathcal{D}(\kappa)})$ as $r\rightarrow0$.
\end{proposition}

The proof proceeds by induction on the dimension $d$. The induction step consists of writing $V_{\kappa}(r)$ as an integral and, using the recursive structure of the BDF $\kappa$, writing this integral as the product of two integrals, splitting into two different cases depending on whether $\kappa$ is outer-max or outer-min.

In the next proposition, we upper-bound a BDF $\kappa$ by a max-norm over a subset of the coordinates that has size exactly $\mathcal{D}(\kappa)$.

\begin{proposition}
    \label{prop:bdfup}
    Let $\kappa : \T^d \rightarrow \R_{\geq 0}$ be a Boolean Distance Function. Then there exists a subset $S \subseteq [d]$ of the coordinates with $|S| = \mathcal{D}(\kappa)$ such that $\kappa(x) \le \max_{i\in S} |x_i|$ for all $x\in \T^d$.
\end{proposition}

The proof is again by induction on the dimension. For the induction step, we have (by the induction hypothesis) two subsets $S_1, S_2 \subseteq [d]$ such that the max-norm on these subsets of coordinates is an upper bound for the comprising functions of $\kappa$: taking $S \coloneqq S_1 \cup S_2$ if $\kappa$ is outer-max, respectively the set of smallest cardinality among $S_1, S_2$ if $\kappa$ is outer-min, completes the induction step.

The next proposition allows us to upper-bound any non-SCOM outer-max BDF by an outer-min BDF acting on the same set of coordinates.

\begin{proposition}
    \label{prop:scomup}
    Let $\kappa : \T^d \rightarrow \R_{\geq 0}$ be a non-SCOM outer-max Boolean Distance Function. Then there exists an outer-min BDF $\kappa' : \T^d \rightarrow \R_{\geq 0}$ such that $\mathcal{D}(\kappa) = \mathcal{D}(\kappa')$ and $\kappa(x) \leq \kappa'(x)$ for all  $x \in \T^d$.
\end{proposition}

The proof also proceeds by induction on $d$, with the key inequality being \[\max(\min(x_1, x_2), \min(x_3, x_4)) \leq \min(\max(x_1, x_3), \max(x_2, x_4)).\]

Combining the three previous propositions, we get the following result, which is one of the main building blocks for the proof of Theorem~\ref{thm:robust}.

\begin{proposition}
    \label{prop:bdfb}
    Let $\kappa : \T^d \rightarrow \R_{\geq 0}$ be a non-SCOM Boolean Distance Function. There exist disjoint subsets $S_1, S_2 \subseteq [d]$ with $\min(|S_1|, |S_2|) = \mathcal{D}(\kappa)$ such that, with $\kappa'((x_i)_{i\in S_1 \cup S_2}) \coloneqq \min(\max_{i\in S_1} |x_i|, \max_{i\in S_2} |x_i|)$, it holds for all $x \in \T^d$ that $\kappa(x) \le \kappa'((x_i)_{i\in S_1 \cup S_2})$. Moreover, $\kappa'$ is a BDF and there exists a constant $c > 0$ such that $V_{\kappa}(r) \le c V_{\kappa'}(r)$ for all $r\ge0$.

\end{proposition}
\begin{proof}
	If $\kappa$ is outer-max, let $\kappa''$ be the outer-min BDF given by Proposition~\ref{prop:scomup} with $\kappa \le \kappa''$ and $\mathcal{D}(\kappa) = \mathcal{D}(\kappa'')$. If $\kappa$ is outer-min, we set $\kappa'' \coloneqq \kappa$ (and notice that we also have $\kappa \le \kappa''$ and $\mathcal{D}(\kappa) = \mathcal{D}(\kappa'')$ in that case). Let $\kappa_1, \kappa_2$ be the comprising functions of $\kappa''$. Applying Proposition~\ref{prop:bdfup} to $\kappa_1$ and $\kappa_2$ we get disjoint subsets $S_1, S_2 \subseteq [d]$ with $|S_k| = \mathcal{D}(\kappa_k)$ such that $\kappa_k(x) \le \max_{i\in S_k} |x_i|$ for $k=1,2$. Let $\kappa'((x_i)_{i\in S_1 \cup S_2}) \coloneqq \min(\max_{i\in S_1} |x_i|, \max_{i\in S_2} |x_i|)$. Then $\kappa(x) \leq \kappa''(x) = \min(\kappa_1(x), \kappa_2(x)) \leq \min(\max_{i\in S_1} |x_i|, \max_{i\in S_2} |x_i|) = \kappa'((x_i)_{i\in S_1 \cup S_2})$ for all $x\in\T^d$. Moreover, since $\kappa''$ is outer-min, we have $\mathcal{D}(\kappa) = \mathcal{D}(\kappa'') = \min(\mathcal{D}(\kappa_1), \mathcal{D}(\kappa_2)) = \min(|S_1|, |S_2|)$ as desired. Finally since $\mathcal{D}(\kappa') = \min(|S_1|, |S_2|) = \mathcal{D}(\kappa)$ we have that $V_{\kappa}(r), V_{\kappa'}(r) \in \Theta(r^{\mathcal{D}(\kappa)})$, and hence there exists $c > 0$ such that $V_{\kappa}(r) \le c V_{\kappa'}(r)$ for all $r\ge0$.
\end{proof}

\section{Small Separators in Single-Coordinate-Outer-Max GIRGs}\label{sec:small_seperators}
Our next main result states that GIRGs induced by BDFs that are SCOM have natural sub-linear separators, which, as it turns out, run along the singled-out coordinate axis. The key proof idea is to partition the ground space along the singled-out coordinate axis into two half-spaces of equal volume, ensuring that each half-space contains a linear number of the vertices of the giant. We then upper-bound the number of edges crossing the separating hyperplanes. Each pair of vertices will contribute an edge intersecting one of the two hyperplanes if and only if the vertices lie in different half-spaces and are connected by an edge. The joint probability of this event can be computed using the law of iterated probability. One can show that the number of crossing edges is $o(n)$ with high probability. 

Thus the two subgraphs can be disconnected through the removal of the $o(n)$ crossing edges, yielding Theorem~\ref{thm:small_separator}.\footnote{The full proof can be found in the Appendix~\ref{sec:proof_small_sep}.} We note here that the max-norm is also a SCOM BDF, and hence our results extend the result proved for max-norms in~\cite{bringmann2019geometric}.

\begin{theorem}
\label{thm:small_separator}
    Let $\kappa : \T^d \rightarrow \R_{\geq 0}$ be a SCOM BDF and let $\mathcal{G} = (\mathcal{V}, \mathcal{E})$ be a $\kappa$-GIRG. Then, w.h.p.\ there exists a subset of edges $\mathcal{S} \subset \mathcal{E}$ with $|\mathcal{S}| = o(n)$ such that $\mathcal{G}' \coloneqq (\mathcal{V}, \mathcal{E}\setminus\mathcal{S})$ has two connected components of size $\Theta(n)$.
\end{theorem}

\section{Robustness of non-Single-Coordinate-Outer-Max GIRGs}\label{sec:robustness}
In this section, we consider GIRGs induced by non-SCOM BDFs and show that they are robust, i.e., they do not contain separators of sub-linear size in the giant component. Our goal is, more precisely, to prove the following theorem:
\begin{theorem}
    \label{thm:robustness}
    Let $\kappa : \T^d \rightarrow \R_{\geq 0}$ be a non-SCOM BDF and let $\mathcal{G} = (\mathcal{V}, \mathcal{E})$ be a $\kappa$-GIRG.
    Then, w.h.p.\ for any subset of edges $\mathcal{S} \subseteq \mathcal{E}$ such that $\mathcal{G}' \coloneqq (\mathcal{V}, \mathcal{E}\setminus\mathcal{S})$ has two connected components of size $\Theta(n)$, it holds that $|\mathcal{S}| = \Omega(n)$.
\end{theorem}

First, we state a lemma that bounds the number of small cuts in connected graphs~\cite{luczak2001bisecting}. This lemma is the inspiration for the proof of robustness of MCD-GIRGs~\cite{lengler2017existence} and will also be used our proof of Theorem~\ref{thm:robustness}.

\begin{lemma}[Lemma 7 in \cite{luczak2001bisecting}]
    \label{lem:cut}
    For any $\varepsilon >0$ there exists $\eta_0(\varepsilon) > 0$ and $n_0$ such that for all $n \geq n_0$, and for all connected graphs $G$ with $n$ vertices, there are at most $(1 + \varepsilon)^n$ many bipartitions of G with at most $\eta_0 n$ cross-edges.
\end{lemma}

We will also need the following notion of a sparse cut.
\begin{definition}\label{def:cut}
    For a graph $\mathcal{G} = (\mathcal{V}, \mathcal{E})$ and constants $\delta,\eta>0$, a \emph{$(\delta,\eta)$-cut} is a partition of $\mathcal{V}$ into two sets of size at least $\delta|\mathcal{V}|$ such that there are at most $\eta|\mathcal{V}|$ \emph{cross-edges}, i.e.\ edges that have one endpoint in each of the sets.
\end{definition}
The proof strategy we follow is almost the same as in~\cite{lengler2017existence}, with some small but crucial modifications.

\subsection{Edge Insertion Criteria} To extend the two-round edge exposure procedure from MCD-GIRGs to general $\kappa$-GIRGs, we will use the upper bounds in terms of an outer-min distance function for the BDF at hand that were established in Section~\ref{sec:bdf}. Our overarching goal will be to expose at first only a subset of the coordinates of each vertex, insert some of the edges based on this partial information, then reveal the remaining coordinates which will lead to the creation of additional edges. More concretely, our aim will be to construct for each pair of vertices $u,v$ a pair of independent random variables $(Y_{uv}^1, Y_{uv}^2)$, inserting edges in the first round if $Y_{uv}^1<p_{uv}$ and in the second round if $Y_{uv}^2<p_{uv}$. The careful modification of the distribution of these two random variables will allow us to emulate in two rounds the one-round sampling procedure where an edge is inserted if $Y_{uv}<p_{uv}$ where $Y_{uv}$ would be drawn uniformly at random from $[0,1]$.\footnote{We refer the reader to \cite{lengler2017existence} for a detailed discussion.} 

Consider therefore a $\kappa_0$-GIRG for some non-SCOM BDF $\kappa_0$. By Proposition \ref{prop:bdfb}, we can find disjoint subsets $S_1, S_2 \subseteq [d]$, with $\mathcal{D}(\kappa_0) = \min(|S_1|, |S_2|)$ such that the outer-min BDF given by $\kappa((x_i)_{i\in S_1 \cup S_2}) \coloneqq \min(\max_{i\in S_1} |x_i|, \max_{i\in S_2} |x_i|) \eqqcolon \min(\norm{(x)_{S_1}}_{\infty}, \norm{(x)_{S_2}}_{\infty})$ is an upper bound for $\kappa_0(x)$. This implies in particular that there exists a $\kappa$-GIRG for some sufficiently small choice of constants $c'_U, c'_L$ such that for every $u, v \in \mathcal{V}$, $p_{\kappa}(u, v) \leq p_{\kappa_0}(u, v)$, where the former denote the connection probabilities of vertices $u$ and $v$ in a $\kappa$-GIRG respectively in a $\kappa_0$-GIRG. Without loss of generality, we will assume that $S_1 \subseteq \{1, .., d - m\}$ and $S_2 = \{d - m + 1, .., d\}$. Additionally, we can assume that $m = |S_2| \leq |S_1|$, which also means that $D \coloneqq \mathcal{D}(\kappa_0) = \mathcal{D}(\kappa) = m$. We will later modify the original algorithm given in~\cite{lengler2017existence} by splitting the sampling process into the first $d - m$ and the last $m$ coordinates (in the original approach by Lengler and Todorovi\'{c}, the split was between the first $d - 1$ coordinates and the last coordinate).

Let $Y_{uv}$ be a uniform random variable over the interval $[0, 1]$, and define two iid random variables $Y_{uv}^1, Y_{uv}^2$ distributed as follows:
\[ \pr{Y_{uv}^1 < c} = \pr{Y_{uv}^2 < c} = 1 - \sqrt{1 - c}.\]
Then the following also hold:
\begin{align*}
    c/2 \leq \pr{Y_{uv}^1 < c} &= \pr{Y_{uv}^2 < c} \leq c, \\
    \pr{\min(Y_{uv}^1, Y_{uv}^2) < c} &= \pr{Y_{uv} < c} = c.
\end{align*}
Let
\[ p_{uv}(c, x) \coloneqq c \cdot \min  \Big\{ 1, \big( \frac{w_u w_v}{n |x|^D}\big)^{\alpha} \Big\},\]
and define the Edge Insertion Criterion (EIC) as
\begin{equation*}\label{eq:eic}
   Y_{uv} \eqqcolon \min(Y_{uv}^1, Y_{uv}^2) < p_{uv}(c_L, \kappa_0(x_u - x_v)). \tag{EIC} 
\end{equation*}
By our choice of $c_L'$, we have the following lower bound:
\[ p_{uv}(c_L, \kappa_0(x_u - x_v)) \ge p_{uv}(c'_L, \kappa(x_u - x_v)) = \max(p_{uv}(c'_L, \norm{(x_u - x_v)_{S_1}}_{\infty}), p_{uv}(c'_L, \norm{(x_u - x_v)_{S_2}}_{\infty})). \]
Thus, we obtain the two following sufficient conditions for edge insertion:
\begin{align}
    Y_{uv}^1 &< p_{uv}(c'_L, \norm{(x_u - x_v)_{S_1}}_{\infty}), \tag{LB1}\label{eq:lb1} \\
     Y_{uv}^2 &< p_{uv}(c'_L, \norm{(x_u - x_v)_{S_2}}_{\infty}). \tag{LB2}\label{eq:lb2}
\end{align}

Crucially notice that if we insert the edges according to~\eqref{eq:lb1} first, we obtain a GIRG (with respect to the max-norm on $\T^{|S_1|}$), which means that it satisfies all the properties mentioned in Section~\ref{sec:girg_properties}. In particular, after the insertion of the first batch of edges, our graph will already contain a (unique) giant component.

We now proceed to describe an algorithm that is central in proving the robustness of a $\kappa_0$-GIRG. It follows closely the algorithm used in~\cite{lengler2017existence}, but we need to modify it to make the proof work in our more general setting. The main modifications are as follows. Firstly, we use the updated~\eqref{eq:eic}, \eqref{eq:lb1} and~\eqref{eq:lb2}. Secondly, instead of first sampling the first $d - 1$ coordinates and then the last coordinate, we first sample the first $d - m$ coordinates and then sample the last $m$ coordinates.

\subsection{Sampling algorithm}\label{sec:algorithm}

In this section we describe the procedure we use for uncovering the edges of the $\kappa_0$-GIRG. We fix some constant $\delta\in(0,1)$. The algorithm can be decomposed into 6 phases. \\ \\
\noindent 
\textbf{Phase 1} We start by sampling $Y^1_{uv}$ for all pairs of vertices $u, v \in \mathcal{V}$. Additionally, for every vertex $u \in \mathcal{V}$, we sample the first $d - m$ coordinates of its position - $(x_{u i})_{1 \leq i \leq d - m}$ - independently and uniformly at random from $[0, 1]$. This is sufficient to determine the graph induced by the edge insertion criterion~\eqref{eq:lb1}, which we refer to as $G_1$. By Theorem \ref{thm:giant}, $G_1$ has a unique linear-sized component (the giant) with at least $s_{max} n$ vertices w.h.p.; we will assume that this holds for the rest of the proof. We denote this giant by $K^1_{max}$. \\ \\
\noindent
\textbf{Phase 2} From \eqref{eq:pl2} we can obtain a constant $B'$ such that at least half the vertices of $K^1_{max}$ have a weight less than $B'$. This can be done by setting $\eta = 1$ and $B' > (2 c_2 / s_{max})^{1 / (\beta - 2)}$. Next, we sub-sample $F' \subseteq \mathcal{V}$ by including every vertex (not just those in the giant) with weight less than $B'$ into $F'$ independently with probability $4f / s_{max}$ for some constant $0 < f < (s_{max} / 12) \cdot \min\{ \delta, s_{max}\}$ to be determined later (in Lemma~\ref{lem:type3}).\\ \\
\noindent
\textbf{Phase 3} Now set $F \coloneqq F' \cap K^1_{max}$. It is straightforward to see that by the choice of parameters it holds that $2fn \le \ex{|F|} \le \E[|F'|] \le 4fn / s_{max}$. Additionally, by Chernoff's bounds (Lemma \ref{lem:chernoff}), we have that $fn \le |F| \le |F'| \le 6fn / s_{max}$ w.h.p.; we will assume that this holds for the rest of the proof.\\ \\
\noindent
The final three phases are split up into $n$ steps in total (one step for each vertex). In each step, we draw the last $m$ coordinates of some vertex and potentially add some incident edges according to the edge insertion criterion~\eqref{eq:eic}. The order in which the vertices are treated is as follows - first the vertices that are not in $K^1_{max}$ (Phase 4), then the remaining vertices that are not in $F$ (Phase 5), and finally the vertices that are in $F$ (Phase 6). Thus, if we order the vertices as $u_1, u_2, \ldots, u_n$ in order to have $K^1_{max} = \{u_i \mid |\mathcal{V} \setminus K^1_{max}| < i \leq n\}$ and $F = \{u_i \mid |\mathcal{V} \setminus F| < i \leq n\}$, then the $k$th step can be described as follows:
\begin{itemize}
    \item Draw $(x_{ki})_{d - m < i \leq m}$ each independently and uniformly at random from $[0, 1]$. (Note that we denote $x_{u_k}$ by $x_k$)
    \item For all $1 \leq j < k$, sample $Y^2_{jk}$ independently.
    \item For all $1 \leq j < k$, add an edge between $u_j$ and $u_k$ if~\eqref{eq:eic} is satisfied. \\
\end{itemize}
\noindent
\textbf{Phase 4} Perform steps $1$ to $|\mathcal{V} \setminus K^1_{max}|$. \\ \\
\noindent
\textbf{Phase 5} Perform steps $|\mathcal{V} \setminus K^1_{max}| + 1$ to $|\mathcal{V} \setminus F|$.
\\ \\
\noindent
\textbf{Phase 6} Perform steps $|\mathcal{V} \setminus F| + 1$ to $n$. \\ \\
\noindent
We denote the resulting graph after Phase $2 + i$ for $i = 2, 3, 4$ by $G_i$ and the corresponding giant component that contains $K^1_{max}$ by $K^i_{max}$. Our aim in the remainder will be to show that the last phase destroys all sublinear cuts which were present in the giant component of $G_3$, while adding only a small number of vertices.

\subsection{Proof of Theorem \ref{thm:robustness}}
We are now ready to prove the main result of this section. The proof is divided into two core steps - first, we show that $K^3_{max}$ has no small cuts in $G_4$. Then we show that the size of $K^4_{max}$ is not much greater than that of $K^3_{max}$. Combining these results, we can show that any small cut in $K^4_{max}$ would induce a small cut in $K^3_{max}$, and hence also exclude the possibility of small cuts in $K^4_{max}$.

First, we observe, recorded in Lemma~\ref{lem:spread},\footnote{A full proof of Lemma~\ref{lem:spread} and the derivation of Corollary~\ref{cor:constant_connection_prob} from it can be found in Appendix~\ref{sec:spread}.} that the neighbors of a given vertex cannot be too concentrated in a small region. To help quantify this we define the notion of cells. For some $M > 0$, we partition $[0, 1]$ into $M$ sub-intervals of equal length $I_j \coloneqq [j/M,(j + 1)/M]$ for $0 \leq j < M$. We define a \emph{$M$-cell} (or just \emph{cell} if $M$ is clear from context) to be a region of the form $\T^{d - m} \times I_{j_{d - m + 1}} \times I_{j_{d - m + 2}} \times \ldots \times I_{j_d}$, where $0 \leq j_{d - m + 1}, j_{d - m + 2}, \ldots, j_d < M$. It is easy to see that the entire space $\T^d$ is partitioned into $M^m$ cells. We call such a partition an \emph{$M$-cell partition}. We also remark that the cell which contains a given vertex $v$ can be completely characterized by the last $m$ coordinates of $v$. 

Lemma~\ref{lem:spread} can now be derived by observing that a set $S$ of $rn$ cells has volume $rn\cdot M^{-m} \in [rl2^{-m}, rl]$ and thus contains an expected number of  $[rnl2^{-m}, rnl]$ vertices. Using the strong Chernoff bound (Lemma~\ref{lem:chernoff}) for a small enough choice of $r$ guarantees that even after a union bound over the possible choices of $S$, with high probability no such set contains at least $\frac{\delta n}{2}$ vertices.

\begin{lemma}
    \label{lem:spread}
    Let $M = \lceil (n / l)^{1/m} \rceil$ for some constant $l\in(0,1]$, and consider the $M$-cell partition of $\T^d$. Then, for every $\delta \in (0,1)$ and every $l$, there exists a constant $r(\delta, l, m) > 0$ such that the following property holds - with probability $1 - e^{-\Theta(n)}$, there is no set $S$ of $rn$ cells (of the considered partition) such that there are at least $\delta n/2$ vertices in the cells of $S$.
\end{lemma}

Intuitively, Lemma~\ref{lem:spread} says that since the positions of vertices are randomly chosen, they must be more or less spread out in the last $m$ coordinates. Thus, we must expect that when we uncover vertices later in the ordering, they are close enough to previously uncovered vertices to have a good chance of edge formation. Indeed, one can show the following useful corollary, which implies a constant lower bound on edge formation in every step of phases 5 and 6. 

\begin{corollary}
    \label{cor:constant_connection_prob}
    Let $\delta\in(0,1)$ be a constant. There is a constant $P > 0$ such that with high probability, the following holds for each step $k > \delta n/2$ of the Algorithm. Let $V_k \coloneqq \{u_i \in \mathcal{V} \mid 1 \leq i < k\}$. For each subset $A \subseteq V_k$ of size at least $\delta n/2$, the probability that step $k$ produces an edge from $u_k$ to $A$ due to~\eqref{eq:lb2} is at least $P$ , i.e.
    \[\pr{\exists v \in A \textnormal{ with } u_kv\in\mathcal{E}} \geq P. \]
\end{corollary}

The remainder of the proof now closely models that of Theorem 3.2 in~\cite{lengler2017existence}.\footnote{Details can be found in Appendix~\ref{sec5details}.} We first prove, using Lemma~\ref{lem:cut}, which restricts the number of sparse bipartitions in a connected graph, that there are no small cuts  in $K_{max}^3$. More precisely, there is a constant $\eta > 0$ such that w.h.p.\ the induced subgraph $G_4[K^3_{max}]$ has no $(\delta, \eta)$-cut (Lemma~\ref{lem:k3cuts}). Then, we can show that in the last phase the giant component does not grow too much, by demonstrating that any "newly added" vertex of $K_{max}^4$ - which contains $K_{max}^3$ entirely - would have to be either in a large non-giant component of $G_3$, or in a non-giant component that contains a vertex of large weight, or in a small component consisting only of small-weight vertices that is nonetheless added to $K^4_{max}$ during this last phase. But each of these categories consists of at most $\delta n$ many vertices (Lemmata~\ref{lem:type1}, \ref{lem:type2}, \ref{lem:type3}), which implies that $|K^4_{max}| \leq |K^3_{max}| + 3 \delta n$ (Lemma~\ref{lem:notbig}). This is the key insight we need to obtain the next result, from which the main theorem of this section follows.

\begin{theorem}
    Let $\delta\in(0,1)$. Then there exists a constant $\eta>0$ such that w.h.p.\ $K^4_{max}$ has no $(4\delta, \eta)$-cuts.
\end{theorem}
\begin{proof}
    Assume that there is a $(4 \delta, \eta)$-cut in $K^4_{max}$. By Lemma \ref{lem:notbig}, such a cut would induce a $( \delta, \eta)$-cut in $K^3_{max}$. But by Lemma \ref{lem:k3cuts}, such a cut cannot exist w.h.p., and hence we get the desired contradiction.
\end{proof}

Since any $\kappa$-GIRG w.h.p.\ has a unique giant component, the above theorem is equivalent to Theorem~\ref{thm:robustness}.

\bibliography{bibliography}

\appendix

\section{Appendix}

\subsection{Tools}
Here we state some very useful lemmata. The first one allows us to sum over (sub)sequences of weights and will be ubiquitous in the proofs.

\begin{lemma}[Lemma 6.3 in~\cite{bringmann2019geometric}]
    \label{lem:sum_to_integral}
    Let $f : \R \rightarrow \R$ be a continuously differentiable function, and recall that $\mathcal{V}_{\ge w} = \{ v \in \mathcal{V} \mid w_v \geq w\}$. Then, for any weights $1 \leq w_0 \leq w_1$, we have
    \begin{align*}
        \sum_{v \in \mathcal{V}, w_0 \leq w_v \leq w_1} f(w_v) = f(w_0) \cdot |\mathcal{V}_{\geq w_0}| - f(w_1) |\mathcal{V}_{\geq w_1}| + \int_{w_0}^{w_1} f'(w) |\mathcal{V}_{\geq w}| dw.
    \end{align*}
\end{lemma}

\noindent Next we give the classical Chernoff bounds.

\begin{lemma}[Theorem 1.1 in~\cite{dubhashi2009concentration}]
    \label{lem:chernoff}
    Let $X \coloneqq \sum_{i = 1}^n X_i$ be the sum of independent indicator random variables $X_i$. Then for any $\varepsilon \in (0,1)$ we have
    \begin{align*}
        \pr{X \geq (1 + \varepsilon)\ex{X}} \leq \exp(-\varepsilon^2\E[X]/3)
        \quad \textnormal{ and } \quad 
        \pr{X \leq (1 - \varepsilon)\ex{X}} \leq \exp(-\varepsilon^2\E[X]/2).
    \end{align*}
\end{lemma}

\noindent We will also need a strong version of Chernoff's bounds.
\begin{lemma}[Theorem 2.15 in~\cite{chung2006complex}]
    \label{lem:strong_chernoff}
    Let $X \coloneqq \sum_{i = 1}^n X_i$ be the sum of independent indicator random variables $X_i$. Then for any $\varepsilon > 0$ we have
    \begin{align*}
        \pr{X \geq (1 + \varepsilon)\ex{X}} \leq \left( \frac{e^{\varepsilon}}{(1 + \varepsilon)^{1 + \varepsilon}} \right)^{\ex{X}} \leq \left( \frac{e}{1 + \varepsilon} \right)^{ (1 + \epsilon) \ex{X}}.
    \end{align*}
\end{lemma}

\subsection{Proofs about Boolean Distance Functions}\label{sec:proof_bdf}

\begin{proof}[Proof of Proposition \ref{prop:volume}]
The proof is by induction on the dimension $d$. Recall that $d$ is a constant w.r.t $n$. When $d=1$, we must have $\kappa(x) = |x|$ and therefore $V_{\kappa}(r) = 2r$ for all $r\in [0, 1/2]$, and in particular $V_{\kappa}(r) = \Theta(r) = \Theta(r^{\mathcal{D}(\kappa)})$. Assume now that $d \ge 2$. Then by Definition~\ref{def:bdf} there exists a non-empty proper subset $S \subsetneq [d]$ of coordinates such that $\kappa(x) = \max(\kappa_1((x_i)_{i \in S}), \kappa_2((x_i)_{i \notin S}))$ or $\kappa(x) = \min(\kappa_1((x_i)_{i \in S}), \kappa_2((x_i)_{i \notin S}))$ for some BDFs $\kappa_1$ and $\kappa_2$. Note that by the induction hypothesis we know that $V_{\kappa_1}(r) = \Theta(r^{\mathcal{D}(\kappa_1)})$ and $V_{\kappa_2}(r) = \Theta(r^{\mathcal{D}(\kappa_2)})$.

\noindent\textbf{Case 1,  $\kappa(x) = \max(\kappa_1((x_i)_{i \in S}), \kappa_2((x_i)_{i \notin S}))$.} In this case, we have
\begin{align*}
    V_{\kappa}(r) &= \int_{x \in \T^d, \kappa(x) \leq r} dx
    = \int_{z \in \T^{d-|S|}, \kappa_2(z) \leq r} \int_{y \in \T^{|S|}, \kappa_1(y) \leq r} dy dz \\
    &= \left(\int_{z \in \T^{d-|S|}, \kappa_2(z) \leq r} dz \right) \cdot \left(\int_{y \in \T^{|S|}, \kappa_1(y) \leq r} dy \right)
    = V_{\kappa_2}(r) \cdot V_{\kappa_1}(r) \\
    &= \Theta(r^{\mathcal{D}(\kappa_2)}) \cdot \Theta(r^{\mathcal{D}(\kappa_1)}) = \Theta(r^{\mathcal{D}(\kappa_1) + \mathcal{D}(\kappa_2)}) = \Theta(r^{\mathcal{D}(\kappa)}),
\end{align*}
where the first two equalities follow from the definition of the Lebesgue measure resp. the case distinction, the third equality follows by the Fubini-Tonelli theorem, and the last equality holds because $\kappa$ is outer-max, and hence $\mathcal{D}(\kappa_1) + \mathcal{D}(\kappa_2) = \mathcal{D}(\kappa)$ by Definition~\ref{def:depth}.

\noindent\textbf{Case 2, $\kappa(x) = \min(\kappa_1((x_i)_{i \in S}), \kappa_2((x_i)_{i \notin S}))$.} In this case, we have
\begin{align*}
    V_{\kappa}(r) &= \int_{x \in \T^d, \kappa(x) \leq r} dx
    = 1 - \int_{x \in \T^d, \kappa(x) > r} dx
    = 1 - \int_{z \in \T^{d-|S|}, \kappa_2(z) > r} \int_{y \in \T^{|S|}, \kappa_1(y) > r} dy dz \\
    &= 1 - \left(\int_{z \in \T^{d-|S|}, \kappa_2(z) > r} dz \right) \left(\int_{y \in \T^{|S|}, \kappa_1(y) > r} dy \right)
    = 1 - \left(1 - V_{\kappa_2}(r) \right) \left(1 -V_{\kappa_1}(r) \right) \\
    &= 1 - \left(1 - \Theta(r^{\mathcal{D}(\kappa_2)}) \right) \left(1 - \Theta(r^{\mathcal{D}(\kappa_1)}) \right) \\
    &= 1 - \left(1 - \Theta(r^{\mathcal{D}(\kappa_1)}) - \Theta(r^{\mathcal{D}(\kappa_2)}) \right)
    = \Theta(r^{\min(\mathcal{D}(\kappa_1), \mathcal{D}(\kappa_2))})
    = \Theta(r^{\mathcal{D}(\kappa)}),
\end{align*}
where  the last equality holds because $\kappa$ is outer-min, and hence $\min(\mathcal{D}(\kappa_1), \mathcal{D}(\kappa_2)) = \mathcal{D}(\kappa)$ by Definition~\ref{def:depth}.
\end{proof}

\begin{proof}[Proof of Proposition \ref{prop:bdfup}]
    The proof is by induction on the dimension $d$. When $d=1$, we must have $\kappa(x) = |x|$ and therefore the claim trivially holds with $S \coloneqq \{1\}$. Assume now that $d\ge 2$, so that $\kappa = \max(\kappa_1, \kappa_2)$ or $\kappa = \min(\kappa_1, \kappa_2)$ for two BDFs $\kappa_1$ and $\kappa_2$ acting on a torus of dimension strictly smaller than $d$. By the induction hypothesis, there exist disjoint subsets $S_1, S_2 \subseteq [d]$ with $|S_1| = \mathcal{D}(\kappa_1), |S_2| = \mathcal{D}(\kappa_2)$ such that $\kappa_1(x^1) \le \max_{i\in S_1} |x^1_i|$ and $\kappa_2(x^2) \le \max_{i\in S_2} |x^2_i|$, where $x^1, x^2$ stand for  $x$ restricted to the coordinates in which $\kappa_1, \kappa_2$ act respectively.
    
    If $\kappa$ is outer-max, setting $S \coloneqq S_1 \cup S_2$ gives us a subset of size $|S| = |S_1| + |S_2| = \mathcal{D}(\kappa_1) + \mathcal{D}(\kappa_2) = \mathcal{D}(\kappa)$. Moreover, we have 
    \begin{align*}
    	 \kappa(x) = \max(\kappa_1(x^1), \kappa_2(x^2)) \le \max\left( \max_{i\in S_1} |x^1_i|, \max_{i\in S_2} |x^2_i|\right) = \max_{i\in S_1 \cup S_2} |x_i|
    \end{align*}
    as desired.

    If $\kappa$ is outer-min, let us assume without loss of generality that $\mathcal{D}(\kappa_1) \le \mathcal{D}(\kappa_2)$, and set $S \coloneqq S_1$. Then $|S| = |S_1| = \mathcal{D}(\kappa_1) = \min(\mathcal{D}(\kappa_1), \mathcal{D}(\kappa_2)) = \mathcal{D}(\kappa)$ and
    \begin{align*}
    	\kappa(x) = \min(\kappa_1(x^1), \kappa_2(x^2)) \le \min\left( \max_{i\in S_1} |x^1_i|, \max_{i\in S_2} |x^2_i|\right) \le \max_{i\in S_1} |x_i|
    \end{align*}
    as desired.
\end{proof}

\begin{proof}[Proof of Proposition \ref{prop:scomup}]
	The proof is by induction on the dimension $d$. Notice that since $\kappa$ is outer-max and non-SCOM, the base case is $d=4$ with $\kappa(x) = \max(\min(|x_1|, |x_2|), \min(|x_3|, |x_4|))$ (up to permutations of the coordinates). We have 
    \begin{align*}
        \kappa(x) &= \max(\min(|x_1|, |x_2|), \min(|x_3|, |x_4|)) \\
        &= \min(\max(|x_1|, |x_3|), \max(|x_1|, |x_4|), \max(|x_2|, |x_3|), \max(|x_2|, |x_4|)) \\
        &\leq \min(\max(|x_1|, |x_3|), \max(|x_2|, |x_4|)) \\
        &\eqqcolon \kappa'(x),
    \end{align*}
    where the second equality follows by distributing the $\max$ operation over the $\min$ operation, and the inequality holds because dropping terms inside the $\min$ does not decrease its value.
    Note that $\kappa'$ is an outer-min BDF, and additionally, $\mathcal{D}(\kappa') = 2 = \mathcal{D}(\kappa)$.
    
    For the induction step, let $d\ge5$ and let $\kappa$ be $\kappa : \T^d \rightarrow \R_{\geq 0}$ be a non-SCOM outer-max Boolean Distance Function with comprising functions $\kappa_1$ and $\kappa_2$. Since $\kappa = \max(\kappa_1, \kappa_2)$ is non-SCOM, we know that $\kappa_1$ is either outer-min or non-SCOM outer-max (if $\kappa_1$ was SCOM, then the coordinate $k$ that can be singled out of $\kappa_1$ could also be singled out of $\kappa$), and the same holds for $\kappa_2$. For $k\in\{1,2\}$, if $\kappa_k$ is non-SCOM outer-max, we know by the induction hypothesis that there exists an outer-min BDF $\kappa_k'$ with $\mathcal{D}(\kappa_k') = \mathcal{D}(\kappa_k)$ and $\kappa_k \le \kappa_k'$. Otherwise if $\kappa_k$ is outer-min we simply set $\kappa_k' \coloneqq \kappa_k$. Let $\kappa_{11}, \kappa_{12}$ be the comprising functions of $\kappa_1'$ and $\kappa_{21}, \kappa_{22}$ be the comprising functions of $\kappa_2'$. Without loss of generality, we can assume that $\mathcal{D}(\kappa_{11}) \le \mathcal{D}(\kappa_{12})$ and $\mathcal{D}(\kappa_{21}) \le \mathcal{D}(\kappa_{22})$. Let us denote by $x^k$ the point $x$ restricted to the coordinates on which $\kappa_k$ acts. We have
    \begin{align*}
        \kappa(x) &= \max(\kappa_1(x^1), \kappa_2(x^2))
        \le \max(\kappa_1'(x^1), \kappa_2'(x^2)) \\
        &= \max(\min(\kappa_{11}(x^{11}), \kappa_{12}(x^{12})), \min(\kappa_{21}(x^{21}), \kappa_{22}(x^{22}))) \\
        &\leq \min(\max(\kappa_{11}(x^{11}), \kappa_{21}(x^{21})), \max(\kappa_{12}(x^{12}), \kappa_{22}(x^{22}))) \\
        &\eqqcolon \kappa'(x).
    \end{align*}
    Clearly $\kappa'$ is an outer-min BDF. It remains to show that $\mathcal{D}(\kappa') = \mathcal{D}(\kappa)$. Indeed
    \begin{align*}
    	\mathcal{D}(\kappa) = \mathcal{D}(\kappa_1) + \mathcal{D}(\kappa_2) = \mathcal{D}(\kappa'_1) + \mathcal{D}(\kappa'_2) = \mathcal{D}(\kappa_{11}) + \mathcal{D}(\kappa_{21})
    \end{align*} 
    and
    \begin{align*}
    	\mathcal{D}(\kappa') = \min(\mathcal{D}(\kappa_{11}) + \mathcal{D}(\kappa_{21}), \mathcal{D}(\kappa_{12}) + \mathcal{D}(\kappa_{22})) = \mathcal{D}(\kappa_{11}) + \mathcal{D}(\kappa_{21}),
    \end{align*} 
    which concludes the proof.
\end{proof}

\subsection{Proof of Theorem \ref{thm:small_separator}}
\label{sec:proof_small_sep}
In the following, we give the full proof of Theorem \ref{thm:small_separator}, which we first restate below for the convenience of the reader.
\begin{theorem}

    Let $\kappa : \T^d \rightarrow \R_{\geq 0}$ be a SCOM BDF and let $\mathcal{G} = (\mathcal{V}, \mathcal{E})$ be a $\kappa$-GIRG. Then, w.h.p.\ there exists a subset of edges $\mathcal{S} \subset \mathcal{E}$ with $|\mathcal{S}| = o(n)$ such that $\mathcal{G}' \coloneqq (\mathcal{V}, \mathcal{E}\setminus\mathcal{S})$ has two connected components of size $\Theta(n)$.
\end{theorem}

\begin{proof}
    Our proof strategy is the same as in~\cite{bringmann2019geometric}. The idea is to partition the ground space along the singled-out coordinate axis into two half-spaces of equal volume, ensuring that each half-space contains a linear number of the vertices of the giant. We then upper-bound the number of edges crossing the separating hyperplanes. Each pair of vertices will contribute an edge intersecting one of the two hyperplanes if and only if the vertices lie in different half-spaces and are connected by an edge. The joint probability of this event can be computed using the law of iterated probability.

    Since $\kappa$ is SCOM, we can assume without loss of generality that it is of the form $\kappa(x) = \max(|x_1|, \kappa'(x_2, \ldots, x_d))$ for $d > 1$, and simply $\kappa(x) = |x_1|$ for $d = 1$. Let $D\coloneqq \mathcal{D}(\kappa)$ denote the depth of $\kappa$. Consider the hyperplanes defined by the equations $x_1 = 0$ and $x_1 = 1/2$. They partition $\T^d$ into two disjoint regions. For any two vertices $u, v$, let us denote by $F_{u, v}$ the event that they lie in different regions of $\T^d$, and by $G^r_{u, v}\coloneqq\{\kappa(x_u - x_v) = r\}$ the event that they are at a distance $r$. By Proposition~\ref{prop:volume}, we know that the probability density $\varrho[G^r_{u, v}]$ of the latter is in $\Theta(r^{D - 1})$, so it remains to estimate $\pr{F_{u, v} \mid G^r_{u, v}}$.

    First, notice that crucially $F_{u, v}$ depends only on the first coordinate of the positions of the vertices $u, v$. To exploit this, let us define the event $H^{r_1}_{u, v} \coloneqq\{|x_{u, 1} - x_{v, 1}| = r_1\}$ for some $0 \leq r_1 \leq r$. The event $F_{u, v}$ is then conditionally independent of the event $G^r_{u, v}$ given $H^{r_1}_{u, v}$. Consequently, for all $0 \leq r_1 \leq r \leq 1/2$, we have $\pr{F_{u, v} \mid G^r_{u, v}, H^{r_1}_{u, v}} = \pr{F_{u, v} \mid H^{r_1}_{u, v}} = 2r_1$. Thus, applying the law of iterated expectation over only the first random coordinate yields 
    \[\pr{F_{u, v} \mid G^r_{u, v}} =  \ex{\pr{F_{u, v} \mid G^r_{u, v}, H^{r_1}_{u, v}}} = \ex{\pr{F_{u, v} \mid H^{r_1}_{u, v}}} = 2r_1 \leq 2r.\]

     Setting $\gamma_{u, v} \coloneqq \min \{ (w_u w_v / n)^{1 / D}, 1/2\}$, we can now express $p_{u, v}$ as
    \begin{align*}
        p_{u, v}(r) = \begin{cases}
        \Theta(1) & \text{ if } r \leq \gamma_{u, v}, \\
        \Theta((\gamma_{u, v} / r)^{\alpha D}) & \text{ if } r > \gamma_{u, v}.
        \end{cases}
    \end{align*} 

    Let $\rho_{u, v}$ be the probability that $u$ and $v$ are connected by an edge that crosses the separating hyperplanes. Then, following the computations that are done in the proof of Lemma 6.1 in~\cite{bringmann2019geometric}, we can show that $\rho_{u, v} \leq O(\gamma_{u, v}^{D + 1}) + O(\gamma_{u, v}^{D \alpha} |\log(\gamma_{u, v})|)$.\footnote{For the sake of readability we defer these computations to subsections \ref{sec:computations_rho} and \ref{sec:computations_S} right after the proof.\label{foot:computations}}
    Defining $\tilde{\alpha} \coloneqq \min(\alpha, 1 + 1/D)$, we get $\rho_{u,v} \le O(\gamma_{u, v}^{D \tilde{\alpha}} \log(n))$ since $\gamma_{u,v} \ge 1/n$. Let $S \coloneqq \sum_{u, v \in \mathcal{V}} \gamma_{u, v}^{D \tilde{\alpha}}$ and let $\mathcal{S}$ denote the (random) set of edges in $\mathcal{G}$ that cross the separating hyperplanes. Then $\E[|\mathcal{S}|] = O(S\log(n))$ (we extracted the log term to ease the reading of the computations to come), and following the computations that are done in the proof of Lemma 6.1 in~\cite{bringmann2019geometric} yields (for any choice of constant $\eta>0$) $S \le O(n^{3-\beta+\eta}) + O(n^{2-\tilde{\alpha}})$.\footnoteref{foot:computations} Defining $m \coloneqq \max\{3 - \beta, 2 - \alpha , 1 - 1/D\} + 2\eta$, and choosing $\eta$ small enough so that $m<1$, we get $\E[|\mathcal{S}|] = O(n^m)$, and therefore, using Markov's inequality, we deduce that $|\mathcal{S}| = o(n)$ w.h.p.

    It remains to show that with high probability $\mathcal{G}'$ has two connected components of linear size. By Chernoff's bounds (Lemma \ref{lem:chernoff}), w.h.p.\ each half-space contains $\Omega(n)$ vertices.  By Lemma 3.12 in \cite{keusch2018geometric} the subgraphs obtained by restricting the vertex set to one of these half-spaces are also GIRGs themselves, and hence by Theorem~\ref{thm:giant} each half-space gives rise to a connected component of size $\Theta(n)$.
    
\end{proof}

\subsubsection{Upper-bounding $\rho_{u,v}$}\label{sec:computations_rho}

Recall that $\rho_{u,v}$ is the probability that $u$ and $v$ are connected by an edge that crosses the
separating hyperplanes. Recall further that $F_{u, v}$ denotes the event that they lie in different regions of $\T^d$, and that $G^r_{u, v}\coloneqq\{\kappa(x_u - x_v) = r\}$ is the event that they are at a distance $r$. More concretely, these events allow us to express $\rho_{u,v}$ as the marginal probability - obtained from integrating over all possible distances $r$ - of the following joint probability: the probability that the pair of vertices $u,v$ is at some distance $r$, with the two vertices not in the same partition and connected by an edge. The possible distances $r$ range from $0$ to $\frac{1}{2}$ due to the choice of the ground space.  Therefore, we have

\begin{align*}
   \rho_{u, v} &= \int_{r = 0}^{1/2} p_{u, v}(r) \cdot \varrho[G^r_{u, v}] \cdot \pr{F_{u, v} \mid G^r_{u, v}} dr \\
   &\leq \int_{r = 0}^{\gamma_{u, v}} 
   \Theta(1) \cdot \Theta(r^{D - 1}) \cdot 2r dr  + \int_{r = \gamma_{u, v}}^{1/2} \Theta((\gamma_{u, v} / r)^{\alpha D}) \cdot \Theta(r^{D - 1}) \cdot  2r dr \\
   &= \underbrace{\int_{r = 0}^{\gamma_{u, v}} \Theta(r^D) dr}_{\eqqcolon I_1}  + \Theta(\gamma_{u, v}^{D \alpha}) \cdot \underbrace{ \int_{r = \gamma_{u, v}}^{1/2} \Theta(r^{D - \alpha D}) dr}_{\eqqcolon I_2}.
\end{align*}
The inequality follows from the bounds described in the main body of the proof. We proceed to compute the integrals $I_1$ and $I_2$. For the first integral we get
\[ I_1 = \int_{r = 0}^{\gamma_{u, v}} \Theta(r^D) dr = \Theta(\gamma_{u,v}^{D + 1}). \]
For the second integral, we distinguish two cases. First, we consider the case when $D - \alpha D \neq -1$, in which case we obtain
\[ I_2 = \int_{r = \gamma_{u, v}}^{1/2} \Theta(r^{D - \alpha D}) dr \le  O(\gamma_{u,v}^{D + 1 - \alpha D}) + O(1),\]
where the first term accounts for the cases where the lower bound of integration dominates, and the second term accounts for the cases where the upper bound of integration dominates.
When $D - \alpha D = -1$ we have
\[ I_2 = \int_{r = \gamma_{u, v}}^{1/2} \Theta(r^{-1}) dr =  O(|\log {\gamma_{u,v}}|). \]
Putting these together, we obtain 
\[ \rho_{u, v} \leq O(\gamma_{u, v}^{D + 1}) + O(\gamma_{u, v}^{D \alpha} |\log(\gamma_{u, v})|) \]
as desired.

\subsubsection{Upper-bounding $S$}\label{sec:computations_S}

Recall that $S=\sum_{u, v \in \mathcal{V}} \gamma_{u, v}^{D \tilde{\alpha}}$ corresponds to the expected number of cross-edges up to a logarithmic factor. We have
\begin{align*}
    S &=
    \sum_{u, v \in \mathcal{V}} \gamma_{u, v}^{D \tilde{\alpha}} = \sum_{u \in \mathcal{V}} \left(\sum_{v \in \mathcal{V}_{\leq n / w_u}} \left(\frac{w_u w_v}{n}\right)^{\tilde{\alpha}} + \sum_{v \in \mathcal{V}_{\geq n / w_u}} \frac{1}{2} \right) \\
    &= \sum_{u \in \mathcal{V}}  \left(\frac{w_u}{n}\right)^{\tilde{\alpha}} \sum_{v \in \mathcal{V}_{\leq n / w_u}} w_v^{\tilde{\alpha}} + \frac{1}{2} \sum_{u \in \mathcal{V}} {|\mathcal{V}_{\geq n / w_u}|} \\
    &\leq \underbrace{\sum_{u \in \mathcal{V}} \left(\frac{w_u}{n}\right)^{\tilde{\alpha}} \int_{1}^{n / w_u} O(n w^{1 - \beta + \eta} w^{\tilde{\alpha} - 1})dw }_{\eqqcolon S_1} + \underbrace{\sum_{u \in \mathcal{V}} {|\mathcal{V}_{\geq n / w_u}|}}_{\eqqcolon S_2},
\end{align*}
where the second equality follows from the minimum in the definition of $\gamma_{u,v}$ and the last inequality holds by Lemma~\ref{lem:sum_to_integral} and equation~\eqref{eq:pl2}. Let us first evaluate $S_2$. It follows from~\eqref{eq:pl2} that there are no vertices of weight larger than $w_{max} = \Theta(n^{1/(\beta - 1 - \eta)})$. Define $w' \coloneqq n / w_{max} = \Theta(n^{(\beta - 2 - \eta)/(\beta - 1 - \eta)})$. Then

\begin{align*}
    S_2 &= \sum_{u \in \mathcal{V}, w_u \geq w'} {|\mathcal{V}_{\geq n / w_u}|} 
    \stackrel{\eqref{eq:pl2}}{=} O\left(\sum_{u \in \mathcal{V}, w_u \geq w'} n \cdot (n / w_u)^{1 - \beta + \eta}\right) \\
    &= O\left(n^{2 - \beta + \eta} \sum_{u \in \mathcal{V}, w_u \geq w'} w_u^{\beta - 1 - \eta}\right).
\end{align*}    
By Lemma~\ref{lem:sum_to_integral} and~\eqref{eq:pl2}, the above becomes, for some $\eta'\in(0,\eta)$,
\begin{align*}
    S_2 &= O\left(n^{2 - \beta + \eta} \left[ (w')^{\beta - 1 - \eta} \cdot n \cdot (w')^{1 - \beta + \eta'} + \int_{w'}^{\infty} w^{\beta - 2 - \eta} n w^{1 - \beta + \eta'} dw \right]\right) \\
    &= O\left(n^{3 - \beta + \eta} \left[ (w')^{\eta' - \eta} + \int_{w'}^{\infty} w^{\eta' - \eta - 1}dw \right]\right) 
    = O(n^{3 - \beta + \eta}).
\end{align*}
To evaluate $S_1$, we consider two cases. If $\tilde{\alpha} + \eta \geq \beta  - 1$, then we have
\begin{align*}
    S_1 &= O\left(n \sum_{u \in \mathcal{V}} \left(\frac{w_u}{n}\right)^{\tilde{\alpha}} \left( \frac{n}{w_u} \right)^{ \tilde{\alpha} - \beta + 1 + \eta} \right) 
    = O\left(n^{2 - \beta + \eta} \sum_{u \in \mathcal{V}} w_u^{\beta - 1 - \eta}\right).
\end{align*}
By Lemma~\ref{lem:sum_to_integral} and~\eqref{eq:pl2} we get, for some $\eta'\in(0,\eta)$,
\begin{align*}
    S_1 &= O\left(n^{2 - \beta + \eta} \left[ n +  \int_{1}^{\infty} w^{\beta - 2 - \eta} n w^{1 - \beta + \eta'} dw \right]\right) 
    = O\left(n^{3 - \beta + \eta} \cdot \int_{1}^{\infty} w^{\eta' - \eta - 1} dw \right) \\
    &= O(n^{3 - \beta + \eta}).
\end{align*}
On the other hand, if $\tilde{\alpha} + \eta < \beta - 1$, then we have
$S_1 = O\left(n^{1 - \tilde{\alpha}} \sum_{u \in \mathcal{V}} w_u^{\tilde{\alpha}}\right)$.
By Lemma~\ref{lem:sum_to_integral} and~\eqref{eq:pl2}, for some $\eta'\in(0,\beta - \tilde{\alpha} - 1)$, we have
\begin{align*}
    S_1 = O\left(n^{1 - \tilde{\alpha}} \left[ n +  \int_{1}^{\infty} w^{\tilde{\alpha}-1} n w^{1 - \beta + \eta'} dw \right]\right) 
    = O\left(n^{2-\tilde{\alpha}} \cdot \int_{1}^{\infty} w^{\tilde{\alpha}-\beta+\eta'} dw \right) = O(n^{2 - \tilde{\alpha}}).
\end{align*}
Summing up, we obtain $S \le O(n^{3-\beta+\eta}) + O(n^{2-\tilde{\alpha}})$ as desired.

\subsection{Proof of Lemma~\ref{lem:spread} and deriving Corollary~\ref{cor:constant_connection_prob} from it}
\label{sec:spread}

\begin{proof}[Proof of Lemma~\ref{lem:spread}]
    
    For simplicity, we assume that $rn$ is an integer. It is straightforward to adapt it to the general case. Fix a set $S$ of $rn$ cells. Their total volume, is given by $rn \cdot M^{-m}$ which lies in the interval $[rl2^{-m}, rl]$. Let $X_v$ be the indicator random variable that denotes if the vertex $v$ is in $S$, so that $\pr{X_v = 1} = rn \cdot M^{-m}$.
    Let $X = \sum_{v \in V} X_v$ be the number of vertices in $S$. We have
    \[ \ex{X} = \sum_{v \in V} \ex{X_v} = rn^2 \cdot M^{-m} \in [rnl2^{-m}, rnl], \]
    and therefore $\frac{\delta}{2}n \geq \frac{\delta}{2} \frac{\ex{X}}{rl}$. Since $X$ is a sum of independent Bernoulli random variables, we can use the strong Chernoff's bound (Lemma \ref{lem:strong_chernoff}) to upper-bound the probability that $S$ has at least $\frac{\delta}{2}n$ vertices. Setting $1 + \varepsilon = \frac{\delta}{2rl}$ yields
    \[ \pr{X \geq \frac{\delta}{2}n} \leq \left(\frac{2erl}{\delta}\right)^{\frac{\delta}{2rl} \cdot \frac{rnl}{2^m}} = \left[ \left(\frac{2erl}{\delta}\right)^{\frac{\delta l}{2^{2m + 1}}} \right]^{\frac{2^mn}{l}}.\]
    Thus, we can choose $r$ small enough so that $(2erl/\delta)^{\delta l/2^{2m+1}} \le 1/3$, which yields
    \[\pr{X \geq \frac{\delta}{2}n} \le 3^{-\frac{2^m n}{l}}.\]
    Now, there are at most $2^{M^m} \leq 2^{\frac{2^m n}{l}}$ choices for $S$, and taking a union bound over all such choices, we get that the probability that there exists a choice of $S$ such that at least $\frac{\delta}{2}n$ vertices are in it, is upper bounded by $(2/3)^{\frac{2^m n}{l}} = e^{-\Theta(n)}$.
\end{proof}

\begin{proof}[Proof of Corollary~\ref{cor:constant_connection_prob}]
   
    Choose $l = 1$, and $M=\lceil n^{1/m} \rceil$ as defined in Lemma~\ref{lem:spread}. Then, for any two vertices $u, v$ in the same cell, we have that $\norm{(x_u - x_v)_{S_2}}^D_{\infty} \leq \left(\frac{l}{n}\right)^\frac{D}{m} = \frac{l}{n}$ (remember that $D = m$), and this implies that $p_{uv}(c'_L, \norm{(x_u - x_v)_{S_2}}_{\infty}) = c'_L$ since $w_u, w_v \ge 1$. Thus, we get that
    \[ \pr{Y_{uv}^2 < p_{uv}(c'_L, \norm{(x_u - x_v)_{S_2}}_{\infty}) \mid u, v \text{ are in the same cell}} \geq c'_L / 2,\]
    and therefore by~\eqref{eq:lb2} this guarantees that $\pr{uv\in\mathcal{E} \mid u, v \text{ are in the same cell}} \geq c'_L / 2$.
    Fix some $\delta n/2 < k \leq n$. Applying Lemma~\ref{lem:spread} to $V_k$ yields that with probability $1 - e^{-\Theta(n)}$ , there is no set $S$ of $rn$ cells such that there are at least $\delta n/2$ vertices $v \in V_k$ in them. In particular, this implies that for every subset $A \subseteq V_k$ of size at least $\delta n/2$, there are more than $rn$ cells that contain at least one vertex of $A$. Thus, with probability at least $rn/M^m \geq rl/2^m$, $u_k$ is in a cell with at least one vertex $v$ of $A$, and hence
    \[ \pr{\exists v \in A \textnormal{ with } u_kv\in\mathcal{E}} \geq rl/2^m \cdot c'_L / 2 \eqqcolon P.\]
    Taking a union bound over all possible values of $k$ concludes the proof.
\end{proof}

\subsection{Proving that $K_{max}^3$ and $K_{max}^4$ have no small cuts in the proof of Theorem~\ref{thm:robustness}}\label{sec5details}

Corollary~\ref{cor:constant_connection_prob}, yielding a constant lower bound on edge formation in Phases 5 and 6, as well as the cut bound from Lemma~\ref{lem:cut}, enable us to show that $K^3_{max}$ has no small cuts. Let us fix some $\delta\in(0,1)$ throughout the section.

\begin{lemma}
    \label{lem:k3cuts}
    There is a constant $\eta > 0$ such that w.h.p.\ the induced subgraph $G_4[K^3_{max}]$ has no $(\delta, \eta)$-cut.
\end{lemma}

\begin{proof}
    At the end of phase $5$, consider bipartitions of $K^3_{max}$. From Lemma \ref{lem:cut}, we know that for all $\varepsilon > 0$ there is some $\eta' > 0$ and $n_0$, such that for all $n \geq n_0$ there are at most $(1 + \varepsilon)^n$ many bipartitions with at most $\eta' n$ many cross edges. Thus, this also holds for any bipartitions of $K^3_{max}$ into two sets $C_1$ and $C_2$, each of size at least $\delta n$.

    Since $F \subseteq K^3_{max}$ satisfies $|F| \geq fn$, without loss of generality we can assume that $|F \cap C_1| \geq fn/2$. Additionally, since $|F| \leq 6fn / s_{max}$ and $f / s_{max} < \delta/12$, we have that $|C_2 \setminus F| \geq \delta n / 2$. For some $v \in F \cap C_1$, let $X_v$ be the indicator random variable that detects if there is an edge from $v$ to some vertex in $C_2 \setminus F$. By Corollary~\ref{cor:constant_connection_prob}, we know that $\ex{X_v} \geq P$. Notice that $(X_v)_{v \in F \cap C_1}$ are independent indicator random variables (because we implicitly condition on the positions of the vertices in $C_2 \setminus F$). So $X \coloneqq \sum_{v \in F \cap C_1}X_v$ is a sum of independent indicator random variables with $\ex{X} \geq P fn/2 \eqqcolon \mu$. Thus, by Chernoff's bound (Lemma \ref{lem:chernoff}) we get that $X < \mu/2$ with probability at most $e^{-\mu / 8}$.

    Finally, choose $\varepsilon>0$ so that $1 + \varepsilon < e^{\mu/(9n)}$, and let $\eta \coloneqq \min(\eta', \mu / (2n))$. By a union bound we get that the probability that there is some bipartition of $K^3_{max}$ with at most $\eta n \leq \eta' n$ edges at the end of Phase 5 that still less than $\mu/2 \geq \eta n$ edges after Phase 6 is upper-bounded by $(1 + \varepsilon)^n e^{-\mu / 8} = e^{- \Theta(n)}$.
\end{proof}

Now, the remaining goal is to show that the size of the giant does not grow too much in Phase $6$. To show this, first notice that $K^3_{max} \subseteq K^4_{max}$. We will classify the vertices outside $K^3_{max}$ into $3$ types and show that each vertex type does not contribute a lot of vertices to $K^4_{max} \setminus K^3_{max}$. First, we consider vertices which are in components that are large but not the giant. Let $s_t \coloneqq 4/(s_{max} P)$, where the constants $s_{max}$ and $P$ are taken from Theorem~\ref{thm:giant} and Corollary~\ref{cor:constant_connection_prob} respectively.

\begin{lemma}
    \label{lem:type1}
    W.h.p.\ $G_3$ has at most $\delta n$ vertices that are in non-giant components of size at least $s_t$.
\end{lemma}

\begin{proof}
    Recall that
    \[ \{ u_k \in \mathcal{V} \mid |\mathcal{V} \setminus K^1_{max}| < k \leq |\mathcal{V} \setminus F|\} = K^1_{max} \setminus F.\]
    Let $k > |\mathcal{V} \setminus K^1_{max}|$. At the end of step $k - 1$, let $A_k$ be the set of vertices that are in a non-giant component of size at least $s_t$. 
    Notice that since we are only uncovering vertices in the giant, $A_k$ is non-increasing in $k$ (for the range of $k$s we consider).
    Let $Z_k$ be an indicator random variable that is $1$ if and only if $A_k \neq A_{k + 1}$.

    Now, if $|A_{|\mathcal{V} \setminus F| + 1}| < \delta n$, we are done. Otherwise, we have that $|A_k| \geq \delta n$ for all $|\mathcal{V} \setminus K^1_{max}| < k \leq |\mathcal{V} \setminus F|$. By Corollary~\ref{cor:constant_connection_prob}, we have that $\pr{Z_k = 1 \mid |A_k| \geq \delta n} \geq P$.
    Let $B^P_k$ be independent indicator random variables such that each one of them is $1$ with probability exactly $P$. Notice that 
    \[ \ex{\sum_{k = |\mathcal{V} \setminus K^1_{max}| + 1}^{|\mathcal{V} \setminus F|} Z_k\ \bigg |\ |A_k| \geq \delta n} \geq \ex{\sum_{k = |\mathcal{V} \setminus K^1_{max}| + 1}^{|\mathcal{V} \setminus F|} B^P_k} > P \cdot s_{max} n / 2 \eqqcolon \mu,\]
    where we have used the fact that $f / s_{max} < s_{max} / 12$ and $|F| \leq 6fn / s_{max}$. By Chernoff's bound (Lemma \ref{lem:chernoff}) we get that $\sum_k B^P_k \geq \mu / 2$ w.h.p., and hence we also have that $\sum_k Z_k \geq \mu / 2$ w.h.p. But since $s_t = 4/(s_{max} P)$ this means that we have removed at least $s_t \cdot \mu / 2 > n$ vertices from $A_{|\mathcal{V} \setminus K^1_{max}| + 1}$, which is a contradiction.
\end{proof}

Now we look at vertices that are in small components containing a large-weight vertex.

\begin{lemma}
    \label{lem:type2}
    Let $B'$ be the constant defined in Phase 2. There exists a constant $B \geq B' > 0$ such that $G_3$ has at most $\delta n$ vertices that are in a component of size at most $s_t$ containing at least one vertex of weight at least $B$. 
\end{lemma}
\begin{proof}
    Notice that for $B \coloneqq \max\{B', (c_2 s_t / \delta)^{1 / (\beta - 2)}\}$ we have that there are at most $\delta n / s_t$ many vertices with weight at least $B$ (using $\eta = 1$ in (\ref{eq:pl2})). Thus, there are at most $\delta n$ vertices that are in a component of size at most $s_t$ containing such a vertex.
\end{proof}

Finally, we take care of the remaining type of vertices bounding the number of edges incident to small-weight vertices that are created in Phase 6. To do so, we will need the following  Azuma-Hoeffding bound with two-sided error events.
\begin{lemma}[Theorem 3.3 in~\cite{bringmann2016average}]
    \label{lem:azuma}
    Let $Z_1, .., Z_m$ be independent random variables over $\Omega_1, .., \Omega_m$. Let $Z = (Z_1, .., Z_m)$, $\Omega = \prod_{k = 1}^m \Omega_k$ and let $g \colon \Omega \mapsto \R$ be measurable with $0 \leq g(\omega) \leq M$ for all $\omega \in \Omega$. Let $\mathcal{B} \subseteq \Omega$ be such that for some $c > 0$ and for all $\omega, \omega' \in \Omega \setminus \mathcal{B}$ that differ in at most $2$ components, the following holds
    \[ |{g(\omega) - g(\omega')}|\leq c.\]
    Then for all $t \ge 2M\pr{\mathcal{B}}$ we have
    \[ \pr{|{g(Z) - \ex{g(Z)}|}\geq t} \leq 2 e^{-\frac{t^2}{32mc^2}} + (2 \tfrac{mM}{c} + 1)\pr{\mathcal{B}} \]

\end{lemma}

We are thus ready to prove the outlined bound.
\begin{lemma}
    \label{lem:type3}
    There is a choice of $0 < f < (s_{max} / 12) \cdot \min\{ \delta, s_{max}\}$ such that w.h.p.\ there are at most $\delta n / s_t$ edges from $F$ to vertices of weight at most $B$, where $B$ is the constant from Lemma~\ref{lem:type2}.
\end{lemma}
\begin{proof}
    Recall from their definitions that $F \subseteq F'$, and that $F'$ is chosen independently from all the edges and only depends on the weights of the vertices. Let $S$ be the set of vertices in the graph with weight at most $B$. Recall that all vertices in $F'$ have a weight of at most $B'$. Since $B' \le B$, we know that $F' \subseteq S$, and morevover we have $|F'| \leq (6fn / s_{max}) \eqqcolon f'n$.

    We may assume that $S = \{1, \ldots, |S|\}$, so that we have an ordering of the vertices (it does not need to be the same ordering as in the sampling algorithm from section~\ref{sec:algorithm}). Consider the random variables $Y_i = (Y_{1 i}, Y_{2 i}, .., Y_{(i - 1) i})$ for all $2 \leq i \leq |S|$, where $Y_{ij}$ are the same random variables defined in~\eqref{eq:eic}. Clearly, all the $Y_i$s are independent, as each $Y_{ij}$ is chosen independently from $[0, 1]$. For a given realization $\omega\in\Omega$ of the $\kappa$-GIRG, let $g(\omega)$ be the number of edges between $F'$ and $S$. Clearly, $g$ depends only on the random variables $\{Y_i\}_{2 \leq i \leq |S|}$ and $\{ x_i\}_{1 \leq i \leq |S|}$. Notice that $0 \leq g(\omega) \leq |F'||S| \leq f' n^2$ for all $\omega\in\Omega$.

    We define a "bad" event $\mathcal{B}$ as
    \[ \mathcal{B} \coloneqq \{ \omega \in \Omega \mid \exists u \in S : \deg(u) \geq 2C \log^2 n\}, \]
    where $C > 0$ is the constant from Lemma~\ref{lem:deg}. We have $\pr{\mathcal{B}} = n^{-\omega(1)}$ by Lemma~\ref{lem:deg}. Moreover, for any $\omega, \omega' \in \Omega \setminus \mathcal{B}$ that differ in at most $2$ components (among the $2|S|-1$ components given by $\{Y_i\}_{2 \leq i \leq |S|}$ and $\{ x_i\}_{1 \leq i \leq |S|}$), we have
    \[ |g(\omega) - g(\omega')| \leq 4C \log^2 n \eqqcolon c\]
    since the outcome of every $x_u$ and $Y_u$ affects at most $2C \log^2 n$ edges if $\omega, \omega' \in \Omega \setminus \mathcal{B}$.

    Furthermore, Lemma~\ref{lem:deg} also implies that the expected degree of every vertex in $F'$ is upper-bounded by $B \cdot C$. Thus, we can also upper-bound the expectation of $g$ as follows
    \[ \ex{g(x_1, \ldots, x_{|S|}, Y_2, \ldots, Y_{|S|})} \leq \ex{\sum_{u \in F'} \deg(u)} \leq BCf'n.\]

    Applying Lemma~\ref{lem:azuma} with $t \eqqcolon 2BCf'n - \ex{g(Z)} \geq BCf'n$ thus yields 
    \[ \pr{g(Z) - \ex{g(Z)} \geq t} \le 2e^{-\tfrac{(BCf'n)^2}{32\cdot2n(4C\log^2 n)^2}} + \big(2\tfrac{2nf'n^2}{4C\log^2 n}+1\big) n^{-\omega(1)} = n^{-\omega(1)}.\]
    For a small enough $f$ such that $2BCf' < \delta / s_t$, we obtain that w.h.p.\  $g(Z) \leq \ex{g(Z)} + t = 2BCf'n < \delta n /s_t$ as desired.
\end{proof}
Therefore, we conclude that $K^4_{max}$ is not much larger than $K^3_ {max}$:
\begin{lemma}
    \label{lem:notbig}
     There is some choice of $0 < f < (s_{max} / 12) \cdot \min\{ \delta, s_{max}\}$ such that w.h.p.\
    \[ |K^4_{max}| \leq |K^3_{max}| + 3 \delta n.\]
\end{lemma}
\begin{proof}
    The claim immediately follows from combining Lemmata~\ref{lem:type1}, \ref{lem:type2} and~\ref{lem:type3}.
\end{proof}

\subsection{Clustering coefficient}\label{sec:clustering}

In the subsequent, we show that the clustering coefficient of any BDF-GIRG is in $\Omega(1)$. We begin with the definition of the clustering coefficient, taken from~\cite{lengler2017existence}.
\begin{definition}
    \label{def:cc}
For a graph $\mathcal{G}=(\mathcal{V},\mathcal{E})$ the clustering coefficient of a vertex $v\in \mathcal{V}$ is defined as
\[
\textsc{cc}(v) \coloneqq \begin{cases}\frac{1}{\binom{\deg(v)}{2}}\cdot \#\{\text{triangles in $\mathcal{G}$ containing $v$}\}, & \text{ if } \deg(v) \geq 2, \\ 0, & \text{ otherwise,}\end{cases}
\]
and the (mean) clustering coefficient of $\mathcal{G}$ is defined as $\textsc{cc}(\mathcal{G}) \coloneqq \frac{1}{|\mathcal{V}|} \sum_{v\in \mathcal{G}} \textsc{cc}(v)$.

\end{definition}
We further require the definition of a stochastic relaxation of the triangle inequality.
\begin{definition}[Definition A.3. in \cite{lengler2017existence}]\label{def:stochtriangle}
Let $\kappa : \mathbb{T}^d \to \mathbb{R}_{\geq 0}$ be a measurable, translation-invariant, and symmetric function inducing a surjective volume function $V:\mathbb{R}_0^+ \to [0,1]$.\footnote{This is fulfilled in our case by choosing $\kappa$ to be a BDF and $V$ to be the volume $V_\kappa(r)$ induced by $\kappa$.} We say that $\kappa$ satisfies a \emph{stochastic triangle inequality} if there is a constant $C>0$ such that the following two conditions hold. 
\begin{enumerate}
\item For every $\varepsilon >0$ let $x_1 = x_1(\varepsilon),x_2 = x_2(\varepsilon)$ be chosen independently and uniformly at random in the $\varepsilon$-ball $\{x \in \mathbb{T}^d \mid \kappa({x}) \leq \varepsilon\}$. Then
\begin{align*}
\liminf_{\varepsilon \to 0} \Pr[\kappa(x_1-x_2) \leq C \varepsilon] > 0.
\end{align*}
\item Moreover,
\begin{align*}
\liminf_{\varepsilon \to 0} \frac{V(\{x \in \mathbb{T}^d \mid \kappa(x) \leq \varepsilon\})}{V(\{x \in \mathbb{T}^d \mid \kappa(x) \leq C \varepsilon\})}> 0.
\end{align*}
\end{enumerate}
\end{definition}
Finally, we need the following theorem from~\cite{lengler2017existence}, which reduces the task of lower-bounding the clustering coefficient of BDF-GIRGs to showing that they satisfy the stochastic triangle inequality.

\begin{theorem}[Theorem A.4 in \cite{lengler2017existence}]\label{thm:clustering_from_stochstic_triangle_ineq}
Consider the GIRG model with a distance function that satisfies the stochastic triangle inequality as described in Definition~\ref{def:stochtriangle}, and let $\mathcal{G}$ be a random instance. Then $\textsc{cc}(\mathcal{G}) = \Omega(1)$ with high probability.
\end{theorem}
Hence we need to prove the following.
\begin{lemma}
    \label{lem:stochastic_triangle_ineq}
    Any BDF $\kappa$ satisfies the stochastic triangle inequality with $C = 2$.
\end{lemma}
\begin{proof}
  
    For the first statement, recall that Proposition~\ref{prop:bdfup} gives us a subset $S \subseteq [d]$ of the coordinates with $|S| = \mathcal{D}(\kappa)$ such that $\kappa(x) \le \max_{i\in S} |x_i|$ for all $x\in \T^d$. For $\varepsilon \leq 1/4$, we have
    \[ \kappa(x_1 - x_2) \leq \max_{i\in S} |x_{1, i} - x_{2, i}| \leq \max_{i\in S} |x_{1, i}| + |x_{2, i}| \leq \max_{i\in S} |x_{1, i}| + \max_{i\in S} |x_{2, i}| \]
    Consequently, $\pr{\kappa(x_1 - x_2) \leq 2 \varepsilon} \geq \pr{\max_{i\in S} |x_{1, i}| \leq \varepsilon} \cdot \pr{\max_{i\in S} |x_{2, i}| \leq \varepsilon}$. Now, since $x_1, x_2$ are chosen uniformly and independently at random from the $\varepsilon$-ball centered at the origin, and $\mathcal{D}(\kappa) = |S|$, we have
    \[\pr{\max_{i\in S} |x_{1, i}| \leq \varepsilon} = \pr{\max_{i\in S} |x_{2, i}| \leq \varepsilon} = \frac{V_{\|\cdot\|_S}(\varepsilon)}{V_{\kappa}(\varepsilon)} = \frac{\Theta(\epsilon^{\mathcal{D}(\kappa)})}{\Theta(\epsilon^{\mathcal{D}(\kappa)})} = \Theta(1),\]
    where the third equality equality comes from Proposition~\ref{prop:volume}. Therefore $\pr{\kappa(x_1 - x_2) \leq 2 \varepsilon} = \Theta(1)$, which implies that the limit must evaluate to a non-zero constant.

    For the second statement, we have, using Proposition~\ref{prop:volume} again,
    \begin{align*}
        \liminf_{\varepsilon \to 0} \frac{V(\{x \in \mathbb{T}^d \mid \kappa(x) \leq \varepsilon\})}{V(\{x \in \mathbb{T}^d \mid \kappa(x) \leq C \varepsilon\})} 
        = \liminf_{\varepsilon \to 0} \frac{V_\kappa(\varepsilon)}{V_\kappa(C\varepsilon)}
        = \liminf_{\varepsilon \to 0}\frac{\Theta(\varepsilon^{\mathcal{D}(\kappa)})}{\Theta((C\varepsilon)^{\mathcal{D}(\kappa)})} = \Theta(1),
    \end{align*}
    and as above we can conclude that the limit is a positive constant.
\end{proof}
Combining Theorem~\ref{thm:clustering_from_stochstic_triangle_ineq} and Lemma~\ref{lem:stochastic_triangle_ineq} gives us the desired result.

\begin{theorem}
    Let $G$ be a GIRG induced by a BDF $\kappa$ acting on the $d$-dimensional torus $\T^d$. Then, with probability $1-o(1)$, its clustering coefficient is constant, i.e.\ $\textsc{cc}(G) = \Theta(1)$. 
\end{theorem}

\end{document}